\documentclass[12pt]{article}

\usepackage{amssymb}
\usepackage{amsmath}
\usepackage{amsfonts}
\usepackage{amsthm}
\usepackage{fullpage}
\usepackage{mathrsfs}
\usepackage{verbatim}
\usepackage{graphicx}
\usepackage[utf8]{inputenc}
\usepackage[colorlinks=true,citecolor=blue]{hyperref}
\usepackage[english]{babel}
\usepackage{tikz-cd}

\newtheorem{Theorem}{Theorem}[section]
\newtheorem{Lemma}[Theorem]{Lemma}
\newtheorem{Corollary}[Theorem]{Corollary}
\newtheorem{Proposition}[equation]{Proposition}

\theoremstyle{definition}
\newtheorem{Definition}[equation]{Definition}
\newtheorem{Example}[equation]{Example}

\newtheorem{Remark}[equation]{Remark}

\numberwithin{equation}{section}

\newcommand{\F}{{\mathbb F}}
\newcommand{\C}{{\mathbb C}}
\newcommand{\Z}{{\mathbb Z}}

\newcommand{\N}{{\mathbb N}}

\newcommand{\mc}[1]{\mathcal{#1}}
\newcommand{\ms}[1]{\mathscr{#1}}

\newcommand{\mtt}[1]{\mathtt{#1}}
\newcommand{\mbf}[1]{\mathbf{#1}}
\newcommand{\qmd}[1]{\mathtt{d}_{\mathtt{#1}}}

\newcommand{\tr}{\text{tr}}

\DeclareMathOperator{\End}{End}

\title{MDS Stabilizer Poset Codes}
\author{Mahir Bilen Can\\
Tulane University, New Orleans, USA\\
Theoretical Sciences Visiting Program,\\ Okinawa Institute of Science and Technology, Japan\\
\texttt{mahirbilencan@gmail.com, mcan@tulane.edu}}

\begin{document}

\maketitle

\begin{abstract}
Poset metrics in the context of stabilizer codes are investigated. 
MDS stabilizer poset codes are defined. 
Various characterizations of these quantum codes are found. 
Methods for producing examples are proposed.  
\\

\\

\end{abstract}

\normalsize

\section{Introduction}

The purpose of our work is to bring together two significant breakthroughs in the theory of error-correcting codes that took place in the 1990s. The first breakthrough, although much less known compared to the second, is the realization of a fundamental concept: the enhancement of the quality of any error-correcting code defined over a finite field through the manipulation of the underlying metric. While the seeds of this concept were sown a few years prior in the earlier works by Niederreiter \cite{Niederreiter1,Niederreiter2}, its formalization only crystallized with the publication of the paper by Brualdi, Graves, and Lawrence~\cite{BGL}. The authors not only generalized the ideas of Niederreiter but also managed to reconcile them with the findings of Rozenblyum and Tsfasman~\cite{RT1997}.
The second significant advancement was Shor's pivotal discovery of a methodology enabling the creation of a 9-qubit repetition code as elucidated in~\cite{Shor1995}. Shortly after Shor's breakthrough, many scientists, including Shor himself, advanced various aspects of quantum error correcting codes. This topic gained further momentum with the realization that quantum error correcting codes emerging as the fixed point space of a suitable group of operators hold exceptional utility~\cite{Gottesman1996, CRSS1997}. 
This was initially perceived in the binary case, but Rains's work~\cite{Rains1999} expedited the rapid expansion of the stabilizer code construction to encompass quantum error correcting codes over nonbinary stabilizer codes~\cite{AK2001}, \cite{KKKS2006}, \cite{KimWalker2004}. 
In the next several paragraphs, we will explain, leaving details to subsequent sections, how we merge the theory of poset metrics and the quantum error correction.

Let $\Z_+$ denote the semigroup of positive integers. 
For $n\in \Z_+$, we fix a poset $\mtt{P}$ whose underlying set is $[n]:=\{1,\dots,n\}$. 
Let $\xi$ be a primitive $p$-th root of unity, where $p$ is a prime number. 
Let $q:=p^m$ for some $m\in \Z_+$.
The finite field with $q$ elements is denoted by $\F_q$. 
For each $a\in \F_q$, we have unitary operators $X(a):\C^q\to \C^q$, $|x\rangle \mapsto | x+a\rangle$ and $Z(a):\C^q\to \C^q$, $|x\rangle \mapsto \xi^{\tr(ax)} | x \rangle$, where $\{ |x \rangle \mid x\in \F_q \}$ is a fixed orthonormal basis for $\C^q$. 
Let $G_n$ denote the finite group generated by the error operators 
$\xi^c X(a_1)Z(b_1)\otimes\cdots \otimes X(a_n)Z(b_n) : (\C^q)^{\otimes n}\to (\C^q)^{\otimes n}$,
where $\{a_1,b_1,\dots, a_n, b_n\}\subset \F_q$ and $c \in \{0,1,\dots, p-1\}$.
The {\em support} of an error operator $g:=\xi^c X(a_1)Z(b_1)\otimes\cdots \otimes X(a_n)Z(b_n) \in G_n$, denoted ${\rm supp}(g)$, is defined as the set of indices $i\in [n]$ such that $X(a_i)Z(b_i) \neq 1$.
We denote by $\mtt{I}(g)$ the smallest lower order ideal of $\mtt{P}$ containing ${\rm supp}(g)$. 
Then the {\em operator $\mtt{P}$-weight} of $g$, denoted $\mtt{wt}_{\mtt{P}}(g)$, is defined to be the cardinality of $\mtt{I}(g)$, 
that is, $\mtt{wt}_{\mtt{P}}(g) := |\mtt{I}(g)|$.
This definition is crucial for introducing the quantum poset codes.
Detectable errors in classical quantum error-correcting codes are discussed in several articles including~\cite{KLV}.
Let $\qmd{P}\in \Z_+$.
If $Q$ detects all errors in $G_n$ of operator $\mtt{P}$-weight less than $\qmd{P}$ but cannot detect an error of operator $\mtt{P}$-weight $\qmd{P}$, then we call $Q$ a {\em quantum $\mtt{P}$-code of minimum weight $\qmd{P}$}.
To concisely state these definitions, we write: {\em $Q$ is a quantum $\mtt{P}$-code with parameters $[[n,K,\mtt{d}_P]]_q$.} 
Also, in complete analogy to the definition of ordinary stabilizer codes, we say that a quantum $\mtt{P}$-code $Q$ is a {\em stabilizer $\mtt{P}$-code} if there exists a subgroup $S$ of $G_n$ such that $Q$ is given by $Q=\{v\in (\C^q)^{\otimes n} : gv=v\text{ for all $g\in S$}\}$.
We now proceed to state the first main result of this article, which will pave the way for developing quantum poset code analogs of the quantum maximum distance separable (MDS) codes, initially introduced by Knill and Laflamme in the binary case~\cite{KnillLaflamme}, and extended by Rains~\cite{Rains1999} to codes over fields of odd characteristic

\begin{Theorem}\label{intro:T1}
Let $\mtt{P}$ be a poset on $[n]$. 
If a stabilizer $\mtt{P}$-code $Q$ has parameters $[[n,K,\mtt{d}_{\mtt{P}}]]_q$, where $K>1$, then the following inequality holds: 
$K \leq q^{n-2\qmd{P}+2}$.
\end{Theorem}

In the classical theory of poset codes, the first example of MDS poset codes was provided by Rosenbloom and Tsfasman in~\cite{RT1997}, where they introduced the analog of Reed-Solomon codes in the $m$-metric (now referred to as the NRT-metric). Later, Skriganov presented an alternative construction in~\cite{Skriganov2001}, while studying the distribution of points within the unit cube. The MacWilliams identities for MDS NRT-metric poset codes were subsequently explored in~\cite{DoughertySkriganov2002}. These works played a pivotal role in motivating the development of the theory of MDS poset codes, which we now introduce in the context of our work.

We call a stabilizer code $Q$ an {\em MDS stabilizer $\mtt{P}$-code} if its dimension $K$ achieves the upper bound specified by the inequality in Theorem~\ref{intro:T1}. The full range of properties of these codes will become apparent once we introduce a further class of stabilizer poset codes. 
To facilitate our discussion, we first extend our framework by reviewing a well-known but crucial connection between stabilizer quantum codes and additive codes.

Let $H:=\F_{q^t}$ for some $t\in \Z_+$. 
We call a subspace $C\subseteq H^s$, where $s\in \Z_+$, an {\em additive code} if it is $\F_q$-linear, but not necessarily $\F_{q^t}$-linear.
Let $\mtt{P}$ be a poset on $[n]$. 
Let $(a_1,\dots,a_n , b_1,\dots, b_n )$ be an element of $\F_q^{2n}$. 
To simplify our notation we denote it by $\underline{\mathbf{ab}}$, where $\mathbf{a}$ represents the first $n$ entries, and $\mathbf{b}$ represents the latter $n$ entries.
Then we define the {\em symplectic $\mtt{P}$-weight} of $\underline{\mathbf{a}\mathbf{b}}$ as follows: 
$wt_{\Delta\mtt{P}} (\underline{\mathbf{a}\mathbf{b}}) := | \mtt{P}_{\mbf{a}}\cup \mtt{P}_{\mbf{b}} |$, 
where $\mtt{P}_{\mbf{a}}$ (resp. $\mtt{P}_{\mbf{b}}$) stands for the smallest lower order ideal of $\mtt{P}$ that contains the indices of all nonzero entries of $\mathbf{a}$ (resp. of $\mathbf{b}$).
It worths pointing out the fact that if $g\in G_n$ is an operator of the form $g=\xi^c X(\mbf{a})Z(\mbf{b})$, then its operator $\mtt{P}$-weight defined earlier is equal to $wt_{\Delta\mtt{P}}(\underline{\mathbf{a}\mathbf{b}})$. 
The {\em trace-symplectic form on $\F_q^{2n}$} is defined by 
$\langle \underline{\mathbf{a}\mathbf{b}}, \underline{\mathbf{a}' \mathbf{b}'}\rangle_s:= \text{tr}(\mbf{b}\cdot \mbf{a}' - \mbf{b}'\cdot \mbf{a}\rangle$. 
Here, the dots in $\mbf{b}\cdot \mbf{a}'$ and $\mbf{b}'\cdot \mbf{a}$ stand for the ordinary dot products of $n$-tuples. 
In this notation, the second main result of our paper is the following assertion.

\begin{Theorem}\label{intro:T2}
Let $Q$ be a stabilizer $\mtt{P}$-code as in Theorem~\ref{intro:T1}. 
Let $S$ be the stabilizer group of $Q$ in $G_n$. Then  
there is a natural map $\varphi : G_n \to \F_q^{2n}$ sending $S$ onto a self-orthogonal $\F_p$-linear $\F_q$-code $C$ of length $2n$. 
Furthermore, the minimum $\mtt{P}$-weight of $Q$ is given by 
\[
\qmd{P} = \min \{wt_{\Delta\mtt{P}} (\underline{\mathbf{a}\mathbf{b}}) :\  \underline{\mathbf{a}\mathbf{b}} \in C^{\perp_s}\setminus C\},
\] 
where $C^{\perp_s}$ is the dual of $C$ with respect to the trace-symplectic form.
\end{Theorem}

For $t\in \Z_+$, we say that a stabilizer $\mtt{P}$-code $Q$ is {\em $\mtt{P}$-pure to $t$} if its stabilizer group $S$ in $G_n$ does not contain any non-scalar error operator whose operator $\mtt{P}$-weight is less than $t$. A stabilizer $\mtt{P}$-code is called {\em $\mtt{P}$-pure} if it is pure to its minimum $\mtt{P}$-weight $\qmd{P}$.
Regarding these notions, we have a theorem which is useful for producing MDS stabilizer poset codes from additive codes. 
To explain, we fix an $\F_q$-vector space basis $\{1,\gamma\}$ for $\F_{q^2}$ so that $\mathbf{v}\in \F_{q^2}^n$ can be written in the form 
$\mathbf{v} = \mathbf{a}1+\mathbf{b}\gamma$ for some $\mathbf{a},\mathbf{b}$ from $\F_q^n$.
The {\em trace-alternating form on $\F_{q^2}^n$} is defined by
\begin{align*}
\langle \mathbf{v} | \mathbf{w} \rangle_a = \text{tr}_{q/p}\left( \frac{ \mathbf{v}\cdot \mathbf{w}^q - \mathbf{v}^q\cdot \mathbf{w}}{\gamma^{q}-\gamma} \right), \quad \text{ where }\quad \mathbf{v},\mathbf{w}\in \F_{q^2}^n.
\end{align*}
The dual of an additive code $D\subset \F_{q^2}^n$ with respect to the trace-alternating form is denoted by $D^{\perp_a}$.
The map $\psi: \F_{q^2}^n \to \F_q^{2n}$, $\mathbf{a}1+ \mathbf{b}\gamma \mapsto \underline{\mathbf{a}\mathbf{b}}$ is isometric 
in the following sense:
$$
\langle  \mbf{a}1+\mbf{b}\gamma , \mbf{a}'1+ \mbf{b}'\gamma \rangle_a = \langle \underline{\mbf{a}\mbf{b}}, \underline{\mbf{a}'\mbf{b}'} \rangle_s
\quad \text{ for every $\mbf{a}1+\mbf{b}\gamma$, $\mbf{a}'1+ \mbf{b}'\gamma$ from $\F_{q^2}^n$}.
$$
For a proof, see~\cite[Lemma 14]{KKKS2006}.
% is an isometry between $(\F_q^{2n}, \langle \cdot,\cdot \rangle_s)$ and $(\F_{q^2}^n, \langle \cdot | \cdot \rangle_a)$. 
Our next theorem provides us with a criteria for deciding when a stabilizer poset code is an MDS stabilizer poset code.

\begin{Theorem}\label{intro:T3}
Let $Q$ be a stabilizer $\mtt{P}$-code with parameters $[[n,K,\qmd{P}]]_q$. 
Let $D=\psi^{-1}\circ \varphi (S)$, where $\psi$ is defined above and $\varphi$ is the map from Theorem~\ref{intro:T2}. 
If $Q$ is $\mtt{P}$-pure, then the following assertions hold: 
\begin{enumerate}
\item[(1)] $Q$ is an MDS stabilizer $\mtt{P}$-code if and only if $D^{\perp_a}$ is an MDS additive $\mtt{P}$-code.
% such that $d_{\mtt{P}}(D^\perp_a)\leq d_{\mtt{P}}(D)$. 
\item[(2)] If $Q$ is an MDS stabilizer $\mtt{P}$-code, and the inequality $n-\mtt{d}_{\mtt{P}} +2 \leq d_{\mtt{P}}(D)$ holds, then $D$ is an MDS additive $\mtt{P}$-code. 
\end{enumerate}
\end{Theorem}

How often do we encounter an MDS stabilizer poset code? 
The answer is ``very often.''
In our next result we show that for each member of a specific family of linear codes in $\F_{q^2}^n$ we have an MDS stabilizer poset code.

\begin{Theorem}\label{intro:T4}
Let $D$ be an $\F_{q^2}$-linear code in $\F_{q^2}^n$ such that $D^{\perp_a} \subseteq D$.
If $k$ denotes the dimension $\dim_{\F_{q^2}} D$, then there exists a poset $\mtt{P}$ on $[n]$ and a $\mtt{P}$-pure MDS stabilizer $\mtt{P}$-code $Q$ with parameters $[ [ n, n-2k, \mtt{d}_{\mtt{P}}] ]_q$.
\end{Theorem}

Let $\mtt{I}$ be a lower order ideal in a poset $\mtt{P}$ on $[n]$. 
Let $x\in \F_{q}^n$. 
The {$\mtt{I}$-ball around $x$}, denoted $\mathbb{B}_{\mtt{I}}(x)$, is the set of vectors $y\in \F_q^n$ such that smallest lower order ideal 
containing the indices of the nonzero entries of $y-x$ is contained in $\mtt{I}$. 
A $\mtt{P}$-code $C$ in $\F_q^n$ is called an {\em $\mtt{I}$-perfect $\mtt{P}$-code} if all $\mtt{I}$-balls centered around the codewords of $C$ are disjoint, and 
$\F_{q}^n = \bigsqcup_{x\in C} \mathbb{B}_{\mtt{I}}(x)$.
According to~\cite[Theorem 2.12]{HyunKim2008} a $k$-dimensional linear $\mtt{P}$-code $C\subsetneq \F_q^n$ is an MDS $\mtt{P}$-code if and only if $C$ is an $\mtt{I}$-perfect code for every ideal $\mtt{I}$ of size $n-k$.
We have a quantum analog of this theorem.
\begin{Definition}
Let $Q$ be a stabilizer poset code with parameters $[[n,K,\qmd{P}]]_q$.
Let $D$ denote the associated additive $\F_{q^2}$-code, which was defined previously. 
If the dual code $D^{\perp_a}$ is an $\mtt{I}$-perfect $\mtt{P}$-code, then we call $Q$ an {\em $\mtt{I}$-perfect stabilizer $\mtt{P}$-code}.
\end{Definition}

\begin{Theorem}\label{intro:T5}
Let $\mtt{P}$ be a poset on $[n]$. 
Let $Q$ be a $\mtt{P}$-pure stabilizer poset code with parameters $[[n,K,\qmd{P}]]_q$.
Then $Q$ is an MDS stabilizer $\mtt{P}$-code if and only if the number $\log_q K$ is an integer, and $Q$ is an $\mtt{I}$-perfect stabilizer $\mtt{P}$-code for every lower order ideal $\mtt{I}$ of such that $|\mtt{I}|=n-\log_q |D^{\perp_a}|$.
\end{Theorem}

We now provide a brief overview of our article. In the following Section~\ref{S:Preliminaries}, we establish our notation and review three chapters in coding theory: classical codes, stabilizer codes, and poset metrics. 
In Section~\ref{S:Additive}, we extend certain results regarding linear poset codes to the domain of additive codes. More specifically, we introduce poset metrics into the theory of additive codes. One of the main results in this section, Theorem~\ref{T:Additive_MDS_codes}, addresses the characterization of MDS additive poset codes in terms of $\mtt{I}$-perfect additive codes. After making another observation on the minimum distances of two nested additive codes, we proceed to our concise Section~\ref{S:Transfer}, where we prove Theorem~\ref{intro:T2}.
The quantum poset code analogue of Singleton's theorem, namely Theorem~\ref{intro:T1} is established in Section~\ref{S:HKS}. Our characterizations of MDS stabilizer poset codes, Theorem~\ref{intro:T3} and Theorem~\ref{intro:T5}, are proved in Section~\ref{S:MDS1}. Finally, to prove our penultimate mentioned result, Theorem~\ref{intro:T4}, we return to the theory of additive poset codes in Section~\ref{S:MDS2}.

\section{Notation and Preliminaries}\label{S:Preliminaries}

Throughout this text, $p$ will represent a prime number. 
Additionally, we will fix a power of $p$, denoted as $q$, defined by $q := p^m$ for some $m\in \Z_+$. 
We will use $\N$ to denote the monoid of nonnegative integers.

As mentioned earlier, the quantum codes of interest in this article are the stabilizer codes. 
The primary objective of our article is to prove a certain inequality that describes the relationship between the parameters of a stabilizer poset code. 
To elucidate the key components of our results, in this preparatory section, we will briefly review:

\begin{enumerate}
\item Classical codes and Singleton's bound.
\item Stabilizer codes.
\item Poset metrics.
\end{enumerate}

\subsection{Classical codes.}\label{SS:Classical}

Let $n\in \Z_+$.
A subset $C$ of $\F_q^n$ is called a {\em ($q$-ary) code of length $n$}.
If $C$ is an $\F_q$-vector subspace of $\F_q^n$, then we call it a {\em linear code of length $n$}.
In this case, the vector space dimension $k:=\dim_{\F_q} C$ is called the {\em dimension of $C$} as a linear code. 
The {\em Hamming metric on $\F_q^n$}, denoted by $\rho_{\mtt{P}_0}$, is defined as follows.
Let $u$ and $v$ be two vectors from $\F_q^n$. 
Then $\rho_{\mtt{P}_0}(u,v)$ is the number of nonzero entries of $u-v$.
In other words, $\rho_{\mtt{P}_0}(u,v)$ is the cardinality of the {\em support} of $u-v$,
that is, the set of indices of the nonzero entries of $u-v$.  
This cardinality is also called the {\em Hamming weight} of $u-v$, denoted $wt_{\mtt{P}_0}(u-v)$.
Then the {\em minimum (Hamming) distance} of $C$ is the number defined by  
\begin{align*}
d_{\mtt{P}_0}(C) := \min_{u\in C} wt_{\mtt{P}_0}(u).
\end{align*}
If the code in question is clear from the context, we will, by a slight abuse of notation, use $d_{\mtt{P}_0}$ instead of $d_{\mtt{P}_0}(C)$. 
All of this can be succinctly expressed by stating that `$C$ is an $[n,k,d_{\mtt{P}_0}]_q$-code.'
\medskip

In~\cite[Theorem 1]{Singleton}, Singleton elucidated a remarkably straightforward yet fundamental relationship between the parameters of a $q$-ary code.
Let $C$ be any nonempty subset of $\F_q^n$.
Let $d_{\mtt{P}_0}$ be the minimum distance of all codewords that are contained in $C$. 
Then Singleton's theorem states that
\begin{align}\label{A:Singleton's}
|C| \leq q^{n-d_{\mtt{P}_0}+1}.
\end{align}
If the inequality in (\ref{A:Singleton's}) is an equality, then $C$ is called a {\em Maximum Distance Separable code} (abbreviated to MDS code).

\begin{Example}
Let $0\leq k \leq q-1$. 
Let $\mathcal{P}_k$ denote the vector space of all polynomials of degree at most $k-1$ in $\F_q[x]$. 
Let $\alpha$ be a generator of $\F_q^*$. 
The {\em Reed-Solomon code of order $k$} is the evaluation code defined by  
\begin{align*}
\mathcal{RS}_k:= \{ (f(1),f(\alpha),\dots, f(\alpha^{q-2}) :\ f \in \mathcal{P}_k\}.
\end{align*} 
It is not difficult to check that the minimum distance of $\mathcal{RS}_k$ is given by $q-k+1$.
Since the dimension of $\mathcal{RS}_k$ is $k$, the Reed-Solomon code is an MDS code.
\end{Example}

\subsection{Stabilizer codes.}\label{SS:CSS}

We fix a primitive $p$-th root of unity, $\xi := e^{2\pi i /p} \in \C$. 
Let $\tr : \F_q \to \F_p$ denote the trace function, $\tr(\alpha) = \sum_{i=0}^m \alpha^{p^i}$. 
This is an $\F_p$-linear map. 
To each $a\in \F_q$, we will associate two unitary operators by using Dirac's bra-ket notation.

First, we fix an orthonormal basis $\{ |x \rangle :\ x\in \F_q \}$ for $\C^q$. 
Then we can apply the field operations of $\F_q$ to the basis elements.
In particular, we have the addition, 
$$
|x\rangle + |a\rangle = |x+a\rangle \quad \text{ for every $\{x,a\}\subseteq \F_q$}.
$$
Now, we define for each $a\in \F_q$ two unitary operators, $X(a):\C^q\to \C^q $ and $Z(a):\C^q\to \C^q$ by linearly extending to $\C^q$ the assignments 
\begin{align*}
X(a)( |x\rangle ) := | x+a\rangle
\qquad\text{and}\qquad
Z(a)( |x\rangle ) := \xi^{\tr(ax)} | x \rangle,\qquad\qquad \text{ where $x\in \F_q$}.
\end{align*}
The set of products, $\mathcal{E}_1 := \{ X(a) Z(b)\ :\ a,b \in \F_q\}$, is called a {\em nice error basis}.
Indeed, it is easy to see that $\mc{E}_1$ is a vector space basis for the space of operators, $\End(\C^q)$.

For $\mathbf{a}:=(a_1,\dots, a_n)\in \F_q^n$, we set 
\begin{align*}
X(\mathbf{a}) := X(a_1)\otimes \cdots \otimes X(a_n) 
\qquad\text{and}\qquad
Z(\mathbf{a}) := Z(a_1)\otimes \cdots \otimes Z(a_n).
\end{align*}
Similarly to the nice error basis $\mathcal{E}_1$ for $\End(\C^q)$, we have a {\em nice error basis for $\bigotimes_{i=1}^n \End(\C^q)$}, defined by 
\begin{align}\label{A:nicebasisforEnd}
\mathcal{E}_n:= \{ X(\mathbf{a}) Z(\mathbf{b}) \ :\ \{\mathbf{a},\mathbf{b}\}\subseteq \F_q^n\}.
\end{align}
Note that, in (\ref{A:nicebasisforEnd}), the product $X(\mathbf{a}) Z(\mathbf{b})$ is understood as $(X(a_1)Z(b_1))\otimes \cdots \otimes (X(a_n)Z(b_n))$. 
The {\em error group of $\mathcal{E}_n$}, denoted by $G_n$, is the finite group defined by  
\begin{align}\label{A:unitaryrep}
G_n := \{ \xi^c X(\mathbf{a}) Z(\mathbf{b}) \ :\ \{\mathbf{a},\mathbf{b}\}\subset \F_q^n,\ c\in \F_p\}.
\end{align} 
It is easy to check that $G_n$ has order $2p^{2m+1}$.
The right hand side of (\ref{A:unitaryrep}) defines naturally a unitary representation for $G_n$.

The elements of $G_n$ are called {\em error operators}. 
Earlier in the introductory section, we introduced the operator $\mtt{P}$-weight for an error operator $g:= \xi^c X(\mathbf{a}) Z(\mathbf{b})$. 
Also, we introduced the symplectic $\mtt{P}$-weight of a pair $(\mathbf{a},\mathbf{b})$. 
It follows from definitions that 
\begin{align*}
\mtt{wt}_{\mtt{P}}(g):= wt_{\Delta\mtt{P}}(\underline{\mathbf{a}\mathbf{b}}).
\end{align*}
We omit the verifications of these these simple facts. 
\medskip

Let $S$ be a subgroup of $G_n$. 
Recall that the {\em stabilizer code associated with $S$} is the quantum code $Q\subseteq (\C^q)^{\otimes n}$ defined by 
$Q:= \bigcap_{g\in S} \{ v\in (\C^q)^{\otimes n} \ | \ gv=v\}$.
In other words, $Q$ is the weight space of weight 1 for the representation of $S$ on $(\C^q)^{\otimes n}$. 
Let $t\in \Z$.
Recall also that a quantum code $Q$ associated with $S$ is said to be {\em $\mtt{P}$-pure to $t$} if $S$ does not contain non-scalar matrices of operator $\mtt{P}$-weight less than $t$. 
If $Q$ is pure to its minimum distance, then $Q$ is said to be {\em $\mtt{P}$-pure}. 
\medskip

Next, we will review a fact that is directly related to the quantum mechanics of our stabilizer codes. 
A quantum code $Q\subset (\C^q)^{\otimes n}$ is said to {\em detect} an error operator $g\in G_n$ if there exists a scalar $\lambda_g \in \C$ 
such that the equality 
\begin{align}\label{A:detectabilitycondition}
\langle c_1 | g | c_2 \rangle = \lambda_g \langle c_1 | c_2 \rangle
\end{align}
holds for every pair of vectors $(c_1,c_2)\in Q\times Q$.
The detectability condition in (\ref{A:detectabilitycondition}) is equivalent to the condition that $P gP=\lambda_g g$, where $P: (\C^q)^{\otimes n} \to Q$ is the orthogonal projection on $Q$. 
We proceed with the assumption that $Q$ is a stabilizer code such that $\dim_\C Q >1$. 
Let $S$ be the stabilizer group of $Q$ in $G_n$. 
Then $Q$ detects an error operator $g\in G_n$ if and only if either $g$ is an element of the semidirect product $Z(G_n)\rtimes S$,
where $Z(G_n)$ is the center of $G_n$, or it does not belong to the centralizer of $S$ in $G_n$, denoted $C_{G_n}(S)$. 
This fact was proven by Ashikhmin and Knill in~\cite[Theorem 3]{AK2001}. 
We note in passing the following well-known facts regarding the stabilizer groups: 
\begin{enumerate}
\item $S$ is an abelian group.
\item $S\cap Z(G_n) = \{1\}$. 
\item $C_{G_n}(S)$ contains $Z(G_n)\rtimes S$. 
\end{enumerate}

%Next, we will review an important theorem regarding stabilizer codes in relation with pairs of classical codes. 
%This theorem is known as the {\em $q$-ary CSS construction}. 
%The acronym comes from the authors Calderbank and Shor~\cite{CalderbankShor1996}, and Steane~\cite{Steane1996}.
%We state its version (but in our notation) that is proven in~\cite{KimWalker2004}. 
%\begin{Lemma}\label{L:CSS} (CSS Code Construction)
%Let $\{n,k_1,k_2,d_1,d_2\}\subseteq \Z_+$. 
%Let $C_1$ and $C_2$ be two linear codes with parameters $[n,k_1,d_1]_q$ and $[n,k_2,d_2]_q$, respectively. 
%If the inclusion $C_1\subseteq C_2$ holds true, then there exists a stabilizer code $C$ with parameters
%\begin{align*}
%[[n, k_2-k_1, \min \{ d_{\mtt{P}_0}(C_2  \setminus C_1), d_{\mtt{P}_0}(C_1^\perp\setminus C_2^\perp)\}]]_q
%\end{align*}
%that is pure to $\min \{d_{\mtt{P}_0}(C_1),d_{\mtt{P}_0}(C_2)\}$.
%\end{Lemma}

\subsection{Poset metrics.}\label{SS:Poset}

Posets can be used to construct interesting error correcting codes in a variety of ways. 
For example, Can and Hibi~\cite{CanHibi2023} use the order polytopes of posets to construct toric codes. 
In this work, we use posets to define a new metric on already existing codes, rather than to construct the underlying vector space.
To begin, we will review the basic theory of posets and poset metrics. 
For more details, we recommend the recent monograph~\cite{FAPP} by Firer, Alves, Panek, and Pinheiro.

A poset is a set with a partial order relation, which we denote by $\leq$. 
In this manuscript, we consider posets on the set $[n]$, which means that the underlying set of the poset is the set $\{1,2,\dots,n\}$.
Let $\mtt{P}$ be a poset on $[n]$. 
A cover relation in $\mtt{P}$ is an ordered pair $(i,j)$ of elements of $\mtt{P}$ such that $i\lneq j$ and $i\leq c  \leq j$ implies that $c = i$ or $c=j$.
The {\em Hasse diagram of $\mtt{P}$}, denoted $\mathscr{H}(\mtt{P})$, is the directed acyclic graph whose vertex set is $\mtt{P}$
and the edge set of $\ms{H}(\mtt{P})$ consists of the cover relations in $\mtt{P}$. 
We implicitly assume that the orientation of the graph is such that for any edge $(i,j)$ of $\ms{H}(\mtt{P})$, we have $i<j$.
A {\em subposet} of $\mtt{P}$ is a subset $\mtt{A}\subseteq \mtt{P}$ and the order relation on $\mtt{A}$ is the restriction of 
the order relation $\leq$ of $\mtt{P}\times \mtt{P}$ to $\mtt{A}\times \mtt{A}$. 
Note that the restriction of the order relation $\leq$ to $\mtt{A}\times \mtt{A}$ is also denoted by $\leq$.

Every poset has a special family of subposets called {\em lower order ideals}, in analogy with the terminology of ring theory.
Let $\mtt{A}$ be a subset of $\mtt{P}$. 
A convenient notation for the set $\{ x\in \mtt{P}:\ x \leq a \text{ for some } a \in \mtt{A} \}$ is $\langle a: a\in \mtt{A} \rangle$.
This set, viewed as a subposet of $\mtt{P}$, is called the {\em lower order ideal generated by $\mtt{A}$ in $\mtt{P}$}.
Notice that $\mtt{A}$ is the smallest lower order ideal containing $\mtt{A}$. 

Next, we define the ``poset weight function.''
Let $\mbf{c}:=(c_1,\dots, c_n)\in \F_q^n$. 
The {\em support ideal} of $\mbf{c}$, denoted $\mtt{P}_{\mbf{c}}$, is the lower order ideal generated by the indices of the nonzero entries of $\mbf{c}$. 
In other words, the support ideal of $\mbf{c}$ is given by 
\begin{align*}
\mtt{P}_{\mbf{c}} := \langle i\in \mtt{P}:\  c_i\neq 0 \rangle. 
\end{align*}
Then the {\em $\mtt{P}$-weight of $\mbf{c}$} is the natural number defined by 
$wt_{\mtt{P}}(\mbf{c}):= |\mtt{P}_{\mbf{c}} |$.

\begin{Example}
Let $\mtt{P}$ denote the poset defined on the set $[8]$ defined by its Hasse diagram depicted in Figure~\ref{F:supportideal}.
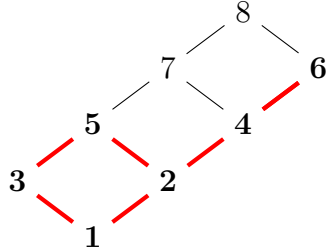
\begin{figure}[htp]
\begin{center}
\begin{tikzpicture}[scale=0.5]
  \node (24) at (0,0) {$8$};
  \node (12) at (-2,-1.5) {$7$};
  \node (8) at (2,-1.5) {{$\mbf{6}$}};
  \node (6) at (-4,-3) {{$\mbf{5}$}};
  \node (4) at (0,-3) {{$\mbf{4}$}};
  \node (3) at (-6,-4.5) {{$\mbf{3}$}};
  \node (2) at (-2,-4.5) {{$\mbf{2}$}};
  \node (1) at (-4,-6) {{$\mbf{1}$}};

  \draw (24) -- (12);
  \draw (12) -- (6);
   \draw[ultra thick, red] (6) -- (3);
  \draw[ultra thick, red] (3) -- (1);
  \draw (24) -- (8); 
  \draw[ultra thick, red] (8) -- (4) -- (2) -- (1);
  \draw (12) -- (4);
  \draw[ultra thick, red] (8) -- (4);
  \draw[ultra thick, red] (6) -- (2);
\end{tikzpicture}
\caption{The Hasse diagram of a poset on $[8]$.}
\label{F:supportideal}
\end{center}
\end{figure}

Let $\mbf{c}:=(0,0,0,0,1,1,0,0) \in \F_q^8$.
Then we have $\mtt{P}_{\mbf{c}}=\langle5,6\rangle= \{ 1,2,3,4,5,6 \}$.
Hence, the $\mtt{P}$-weight of $\mbf{c}$ is given by $wt_{\mtt{P}} (\mbf{c}) =| \{ 1,2,3,4,5,6 \}|=6$.

\end{Example}

The idea of using poset metrics was discovered by Brualdi, Graves, and Lawrance in~\cite{BGL}. 

\begin{Definition}
Given a poset $\mtt{P}$ on $[n]$, the {\em poset metric on $\F_q^n$} is defined by
\begin{align*}
\rho_{\mtt{P}} (\mbf{u}, \mbf{v}) := wt_{\mtt{P}}(\mbf{u} - \mbf{v}) \quad (\mbf{u}, \mbf{v} \in \F_q^n).
\end{align*}
A {\em poset code of length $n$} is a subset $C \subset \F_q^n$ together with the restriction of the poset metric $\rho_{\mtt{P}}$ to $C$.
In this case, we say that $C$ is a {\em $\mtt{P}$-code of length $n$}.
The {\em minimum weight} of a $\mtt{P}$-code $C$ is 
$$
d_{\mtt{P}}(C) :=\min \{ wt_{\mtt{P}}(\mbf{u})\mid \mbf{u} \in \F_q^n\setminus \{0\}\}.
$$
Note: Later in this article we will use the minimum weight of a code $D\subseteq \F_q^{2n}$ with respect to the symplectic $\mtt{P}$-weight, $wt_{\Delta \mtt{P}}$
that is defined in the introduction.
This minimum weight is defined as $d_{\Delta\mtt{P}}(D) :=\min \{wt_{\Delta\mtt{P}}(\underline{\mbf{a}\mbf{b}}) \mid  \underline{\mbf{a}\mbf{b}} \in \F_q^{2n}\setminus \{0\}\}$.
\end{Definition}

\begin{Example}
If $\mtt{P}_0$ denotes the antichain poset on $[n]$, then $\rho_{\mtt{P}_0}$ is nothing but the Hamming distance on $\F_q^n$. Hence, every classical code in $\F_q^n$ is a $\mtt{P}_0$-code of length $n$.
\end{Example}

Many important properties of ordinary error-correcting codes naturally extend to the setting of poset codes. 
For example, in~\cite[Proposition 2.1]{HyunKim2008}, Hyun and Kim generalized the Singleton's theorem (\ref{A:Singleton's}) to poset codes.
We will refer to this result as the {\em Hyun-Kim-Singleton Theorem}.
\begin{Lemma}\label{L:HKSbound}
Let $\mtt{P}$ be a poset on $[n]$.
Let $C$ be a nonempty subset of $\F_q^n$. 
Then we have 
\begin{align}\label{A:HyunKimSingleton}
|C| \leq q^{n-d_{\mtt{P}}(C)+1}.
\end{align}
\end{Lemma}

Following Hyun and Kim, we call a $\mtt{P}$-code $C$ in $\F_q^n$ an {\em MDS $\mtt{P}$-code} if $|C|$ meets the upper bound in (\ref{A:HyunKimSingleton}). 

\begin{Definition}
Let $C$ be a subset of $\F_q^n$. The {\em $q$-dimension} of $C$ is the real number $\log_q |C|$. 
\end{Definition} 

\begin{Remark}\label{R:qdimisinteger}
The $q$-dimension of a code need not be an integer, and can even be negative for non-linear codes.
The $q$-dimension of a linear code in $\F_q^n$ agrees with its $\F_q$-vector space dimension. 
For a non-linear MDS $\mtt{P}$-code in $\F_q^n$, the $q$-dimension is still an integer, namely $n-d_{\mtt{P}}(C)+1$.
\end{Remark}

At the center of the theory of MDS $\mtt{P}$-codes is the notion of an $\mtt{I}$-ball, where $\mtt{I}$ stands for a lower order ideal of $\mtt{P}$.
\begin{Definition}
Let $\mtt{P}$ be a poset on $[n]$. Let $\mtt{I}$ be a lower order ideal in $\mtt{P}$. 
Let $\mbf{u}\in \F_{q}^n$. 
The {$\mtt{I}$-ball around $\mbf{u}$} is defined as the following set:
$$
\mathbb{B}_{\mtt{I}}(\mbf{u}):= \{ \mbf{v}\in \F_q^n\ :\  \mtt{P}_{\mbf{y}-\mbf{v}}\subseteq \mtt{I} \}.
$$
If $\mbf{u}$ is the zero-vector, then we write $\mathbb{B}_{\mtt{I}}$ instead of $\mathbb{B}_{\mtt{I}}(\mbf{u})$.
\end{Definition}

\begin{Lemma}(See \cite[Proposition 2.7]{HyunKim2008})\label{L:4parts}
Let $\mtt{P}$ be a poset on $[n]$. Let $\mtt{I}$ be a lower order ideal in $\mtt{P}$. 
Then the following assertions hold: 
\begin{enumerate}
\item[(a)] $\mathbb{B}_{\mtt{I}}$ is a vector subspace of $\F_q^n$ of dimension $|\mtt{I}|$. 
\item[(b)] For $\mbf{u}\in \F_q^n$, $\mathbb{B}_{\mtt{I}}(\mbf{u})$ is the coset of $\mathbb{B}_{\mtt{I}}$ containing $\mbf{u}$, that is, $\mathbb{B}_{\mtt{I}}(\mbf{u}) = \mbf{u}+ \mathbb{B}_{\mtt{I}}$.
\item[(c)] For $\mbf{u},\mbf{v}\in \F_q^n$, the $\mtt{I}$-balls $\mathbb{B}_{\mtt{I}}(\mbf{u})$ and $\mathbb{B}_{\mtt{I}}(\mbf{v})$ are either identical or disjoint. 
Moreover, 
$$
\mathbb{B}_{\mtt{I}}(\mbf{u}) = \mathbb{B}_{\mtt{I}}(\mbf{v}) \iff \mtt{P}_{\mbf{u}-\mbf{v}}\subseteq \mtt{I}.
$$
\item[(d)] The full space $\F_q^n$ can be written as a disjoint union of $\mtt{I}$-balls.
\end{enumerate}
\end{Lemma}

\section{Additive Poset Codes}\label{S:Additive}

In this section, we bring poset metrics to a setting of additive codes.
Here, by an additive code, we mean a subgroup $D$ of either $\F_q^{2n}$ or $\F_{q^2}^n$. 
This idea is the main engine of our paper since most results about stabilizer codes are obtained through additive codes.
Note that additive codes need not be $\F_q$- or $\F_{q^2}$-linear subspaces, but they are always $\F_p$-linear in their ambient vector spaces. 
Following~\cite{Huffman2013}, we refer to these codes as {\em $\F_p$-linear $\F_q$-codes of length $2n$} or {\em $\F_p$-linear $\F_{q^2}$-codes of length $n$}, respectively.
Nevertheless, to simplify terminology, in some places in the text, we simply write``$D$ is an additive code.''
\medskip

\begin{comment}
Let $\mtt{P}$ and $\mtt{Q}$ be two posets. 
We define the {\em sum of $\mtt{P}$ and $\mtt{Q}$} following~\cite[Section 3.2]{EC1}:
$\mtt{P}+\mtt{Q}$ is the poset defined by $s\leq t$ 
if and only if either both $s$ and $t$ are from $\mtt{P}$ and $s\leq t$ holds in $\mtt{P}$, or both $s$ and $t$ are from $\mtt{Q}$ and $s\leq t$ holds in $\mtt{Q}$.  
It is obvious that $\mtt{P}+\mtt{Q}$ could also be called the disjoint union poset.
Let $r\in \Z_+$. 
Let $\mtt{P}$ be a poset.
We write $r\mtt{P}$ to denote the $r$-fold sum $\mtt{P}+\cdots + \mtt{P}$.

The following simple fact will be useful.
\begin{Lemma}
Let $r\in \Z_+$. 
Let $\mtt{P}$ be a poset on $[n]$ for some $n\in \Z_+$. 
For $\{\mbf{a}_1,\dots, \mbf{a}_r\}\subseteq \F_q^n$, let $\underline{\mbf{a}_1\dots \mbf{a}_r}$ denote the 
vector in $\F_q^{rn}$ obtained by the concatenation of $\mbf{a}_1,\dots, \mbf{a}_r$ in the written order.
Then the $r\mtt{P}$-weight  of $\underline{\mbf{a}_1\dots \mbf{a}_r}$ is given by
$$
wt_{r\mtt{P}}(\underline{\mbf{a}_1\dots \mbf{a}_r})= \sum_{i=1}^r wt_{\mtt{P}}(\mbf{a}_i) = 
\sum_{i=1}^r | \mtt{P}_{\mbf{a}_i}|. 
$$
\end{Lemma}

\end{comment}

Recall from the introduction that the symplectic $\mtt{P}$-weight on $\F_q^{2n}$ is defined by $wt_{\Delta\mtt{P}} (\underline{\mathbf{a}\mathbf{b}}) := | \mtt{P}_{\mbf{a}}\cup \mtt{P}_{\mbf{b}} |$.
Let $\rho_{\Delta \mtt{P}} : \F_q^{2n}\times \F_q^{2n} \to \N$ be the map defined by 
$$
\rho_{\Delta\mtt{P}} ( \underline{\mbf{a}\mbf{b}}, \underline{\mbf{a}'\mbf{b}'} ) :=
wt_{\Delta\mtt{P}}( \underline{\mbf{a}\mbf{b}}-\underline{\mbf{a}'\mbf{b}'} )\qquad 
\text{ for every } \underline{\mbf{a}\mbf{b}}, \underline{\mbf{a}'\mbf{b}'} \in \F_q^{2n}.
$$
We claim that $\rho_{\Delta \mtt{P}}$ is a metric on $\F_q^{2n}$.
Indeed, for $\rho_{\Delta\mtt{P}}$, the symmetry and the positive-definiteness are easy to check. 
The fact that it satisfies the triangle inequality is a direct consequence of the following lemma. 
\begin{Lemma}
Let $\mtt{P}$ be a poset on $[n]$.
Then for every $\underline{\mbf{a}\mbf{b}}$ and $\underline{\mbf{a}'\mbf{b}'}$ from $\F_q^{2n}$, we have 
$$
wt_{\Delta \mtt{P}}( \underline{\mbf{a}\mbf{b}}+\underline{\mbf{a}'\mbf{b}'} ) \leq wt_{\Delta \mtt{P}}( \mbf{a}\mbf{b})+wt_{\Delta \mtt{P}}(\underline{\mbf{a}'\mbf{b}'}).
$$
\end{Lemma}

\begin{proof}
Let $(a_1,\dots, a_n)$ (resp. $(a_1',\dots, a_n'),(b_1,\dots, b_n),(b_1',\dots, b_n')$) be the coordinates of $\mbf{a}$ (resp. of $\mbf{a}'$, 
$\mbf{b}$, $\mbf{b}'$).
Recall that $\mtt{P}_{\mbf{a}}$ denotes the lower order ideal generated by $\{i\in [n] \mid a_i\neq 0\}$ in $\mtt{P}$. 
It is easy to check from this definition that 
$$
\mtt{P}_{\mbf{a}+\mbf{a}'}\subseteq \mtt{P}_{\mbf{a}}  \cup \mtt{P}_{\mbf{a}'} \quad \text{ and }\quad 
\mtt{P}_{\mbf{b}+\mbf{b}'}\subseteq \mtt{P}_{\mbf{b}}  \cup \mtt{P}_{\mbf{b}'}.
$$
Then we see that 
\begin{align*}
wt_{\Delta \mtt{P}} (\underline{\mbf{a}\mbf{b}}+\underline{\mbf{a}'\mbf{b}'} ) &= 
wt_{\Delta \mtt{P}} (\underline{(\mbf{a}+\mbf{a}')(\mbf{b}+\mbf{b}')}) \\
&= | \mtt{P}_{\mbf{a}+\mbf{a}'}\cup \mtt{P}_{\mbf{b}+\mbf{b}'}| \\
&= |\mtt{P}_{\mbf{a}}  \cup \mtt{P}_{\mbf{a}'}  \cup \mtt{P}_{\mbf{b}}  \cup \mtt{P}_{\mbf{b}'}|\\ 
&\leq |\mtt{P}_{\mbf{a}} \cup \mtt{P}_{\mbf{b}}|+ | \cup \mtt{P}_{\mbf{a}'} \cup \mtt{P}_{\mbf{b}'}| \\
&= wt_{\Delta \mtt{P}}( \mbf{a}\mbf{b})+wt_{\Delta \mtt{P}}(\underline{\mbf{a}'\mbf{b}'}).
\end{align*}
This finishes the proof of our assertion.
\end{proof}

We fix a vector space basis $\{1,\gamma\}$ for $\F_{q^2}$ over $\F_q$.
It is straightforward to check that the map $\psi: \F_{q^2}^n\to \F_q^{2n}$,
$\mbf{a}1 + \mbf{b} \gamma \mapsto \underline{\mbf{a}\mbf{b}}$,
which we mentioned in the introduction, is an $\F_q$-linear isomorphism. 
It also preserves poset related metrics.

\begin{Lemma}\label{L:isometry}
The map $\psi: \F_{q^2}^n \to \F_q^{2n}$, $ \mathbf{a}+ \gamma \mathbf{b}\mapsto \underline{\mathbf{a}\mathbf{b}}$ is an $\F_q$-linear isometry between $(\F_{q^2}^n,\rho_{\mtt{P}})$ and $(\F_{q}^{2n}, \rho_{\Delta\mtt{P}})$,
where $\rho_{\Delta\mtt{P}}$ is the metric defined by the symplectic $\mtt{P}$-weight $wt_{\Delta\mtt{P}}$ and $\rho_{\mtt{P}}$ is the ordinary $\mtt{P}$-metric.
\end{Lemma}
\begin{proof}
The fact that $\psi$ is an $\F_q$-linear isomorphism is evident. 
To show that it is an isometry we will show that it preserves weights. 
Let $\mbf{x}$ be a vector from $\F_{q^2}^n$. 
We write $\mbf{x}$ in the form $\mbf{a}1 + \mbf{b} \gamma \in \F_{q^2}^n$.
Let us express $\mbf{a}$ and $\mbf{b}$ in coordinates by $(a_1,\dots, a_n)\in \F_q^n$ and $(b_1,\dots, b_n)\in \F_q^n$, respectively. 
Then the coordinates of $\mbf{x}$ are given by 
$$
\mbf{x}:=(x_1,\dots, x_n) = (a_1 1 + b_1\gamma, \dots, a_n 1 + b_n \gamma).
$$
Let $i\in [n]$. 
Clearly, the $i$-th coordinate of $\mbf{x}$ is nonzero if and only if either $a_i\neq 0$ or $b_i\neq 0$. 
This means that the lower order ideal $\mtt{P}_{\mbf{x}}$ is given by the union of the lower order ideals,
$$
\mtt{P}_{\mbf{x}} = \mtt{P}_{\mbf{a}}\cup \mtt{P}_{\mbf{b}}.
$$
In other words, the $\mtt{P}$-weight of $\mbf{x}$ is exactly the symplectic $\mtt{P}$-weight of $\underline{\mbf{a}\mbf{b}}$ in $\F_q^{2n}$. 
This finishes the proof of our assertion.
\end{proof}

We proceed to extend the definition of an $\mtt{I}$-perfect code to the setting of additive codes.

\begin{Definition}
Let $\mtt{I}$ be an ideal of a poset $\mtt{P}$ on $[n]$. 
An additive code $D\subseteq \F_q^n$ is called {\em $\mtt{I}$-perfect} if the $\mtt{I}$-balls centered around the codewords of $D$ are pairwise disjoint, and their union is $\F_q^n$. 
\end{Definition}

A careful examination of the proof of~\cite[Theorem 2.12]{HyunKim2008} reveals the fact that the corresponding assertion for MDS additive codes holds true with an additional assumption on the cardinality of the code. 
\begin{Theorem}\label{T:Additive_MDS_codes}
Let $\mtt{P}$ be a poset on $[n]$.
Let $D$ be an additive code in $\F_q^n$. 
Then $D$ is an additive MDS $\mtt{P}$-code if and only if the following two conditions are satisfied:
\begin{enumerate}
\item[(1)] $\log_q |D|$ is an integer, and
\item[(2)] $D$ is an $\mtt{I}$-perfect code for every ideal $\mtt{I}$ of size $n-\log_q |D|$. 
\end{enumerate}
\end{Theorem}

\begin{proof}
($\Rightarrow$) 
Assume that $D$ is an additive MDS $\mtt{P}$-code. 
Then, as we mentioned in Remark~\ref{R:qdimisinteger}, the $q$-dimension of $D$ is the integer 
\begin{align}\label{A:define_k}
k:=n- d_{\mtt{P}}+1.
\end{align}
This proves (1).
To prove (2), we will use the arguments of the first part of the proof of~\cite[Theorem 2.12]{HyunKim2008}.

Let $\mtt{I}$ be an ideal of size $n-k$, where $k$ is as in (\ref{A:define_k}).
By Lemma~\ref{L:4parts} (d), we know that $\F_q^n$ can be partitioned into $\mtt{I}$-balls.
The number of $\mtt{I}$-balls is $q^{n- |\mtt{I}|} = q^k = |D|$.
Since the minimum distance of $D$ with respect to $\mtt{P}$ is $n-k+1$, for every distinct codewords $x$ and $y$ from $D$, 
we have $\rho_{\mtt{P}}(x,y) \geq n-k+1 > |\mtt{I}|$. 
Hence, by Lemma~\ref{L:4parts} (c), we see that the balls $\mathbb{B}_{\mtt{I}}(u)$ and $\mathbb{B}_{\mtt{I}}(v)$ are disjoint. 
Since there are in total $|D|q^{|\mtt{I}|}= q^k q^{n-k}=q^n$ elements in the union $\bigsqcup_{x\in D} \mathbb{B}_{\mtt{I}}(x)$,
the union $\bigsqcup_{x\in D} \mathbb{B}_{\mtt{I}}(x)$ must equal to $\F_q^n$.
Therefore, we proved that $D$ is an $\mtt{I}$-perfect code. 

\medskip
($\Leftarrow$)
This direction of the proof is similar to the second part of the proof of~\cite[Theorem 2.12]{HyunKim2008}.
Let $D$ be an additive $\mtt{P}$-code such that  
1) the real number $k$ defined by $k:=n-\log_q |D|$ is an integer, and 2) $D$ is an $\mtt{I}$-perfect code for every ideal $\mtt{I}$ of size $k$.
Assume towards a contradiction that there are two elements $x$ and $y$ in $D$ such that $\rho_{\mtt{P}}(x,y) \leq n-k$. 
Then $x$ and $y$ belong to the same $\mtt{J}$-ball such that $|\mtt{J}|\leq n-k$. 
It follows that there is an order ideal $\mtt{I}$ of cardinality $n-k$ such that $\{x,y\}\subseteq \mtt{J}\subseteq \mtt{I}$. 
But this means that $D$ cannot be an $\mtt{I}$-perfect code, contradicting with the hypothesis of the implication ($\Leftarrow$). 
In other words, we must have $\rho_{\mtt{P}}(x,y) > n-k$.
Hence, the minimum distance of $D$ with respect to $\mtt{P}$ is at least $n-k+1$.
We conclude from Lemma~\ref{L:HKSbound} that the minimum distance of $D$ with respect to $\mtt{P}$ must be exactly $n-k+1$.
Hence, $D$ is an additive MDS $\mtt{P}$-code in $\F_q^n$. 
This finishes the proof of our theorem.
\end{proof}

In~\cite{CRSS1998}, for additive codes in $\F_4^n$, a notion of `pureness' was introduced. 
We have a relevant definition.
\begin{Definition}\label{D:pureness}
We say that an additive $\mtt{P}$-code $C\subseteq \F_{q^t}^n$ is {\em $\mtt{P}$-pure} with respect to an inner product on $\F_{q^t}^n$ if the dual code 
$C^{\perp}$ does not contain a vector of $\mtt{P}$-weight less than $d_{\mtt{P}}(C)$.
In other words, $C$ is {\em $\mtt{P}$-pure} if the inequality $d_{\mtt{P}}(C) \leq d_{\mtt{P}}(C^{\perp})$ holds.
More specifically, for an additive code $D\subseteq \F_{q^2}^{n}$, we say that {\em $D$ is $\mtt{P}$}-pure if the inequality 
$d_{\mtt{P}}(D) \leq d_{\mtt{P}}(D^{\perp_a})$ holds.
Similarly, if $C$ is an additive code in $\F_{q}^{2n}$, and the inner product is given by the trace-symplectic form, then 
$C$ is $\Delta\mtt{P}$-pure if $d_{\Delta\mtt{P}}(C)\leq d_{\Delta\mtt{P}}(C^{\perp_s})$ holds. 
\end{Definition}

\begin{Remark}\label{R:Pure}
In light of Lemma~\ref{L:isometry}, we see that a code $D\subseteq \F_{q^2}^n$ is $\mtt{P}$-pure if and only if the corresponding additive code 
$C$ in $\F_q^{2n}$ is $\Delta\mtt{P}$-pure. 
It follows that, if $Q$ is a stabilizer $\mtt{P}$-code, then $Q$ is $\mtt{P}$-pure if and only if the corresponding additive code $D$ is $\mtt{P}$-pure. 
\end{Remark}
The following lemma gives us a helpful criteria for deciding when a self-orthogonal code $C\subset \F_q^{2n}$ is a $\mtt{P}$-pure additive code. 

\begin{Lemma}\label{L:morethantwo} 
Let $\mtt{P}$ be a poset on $[n]$. 
Let $C\subseteq D$ be two $\F_q$-linear $\F_{q^t}$-codes of length $n$ such that the quotient $\F_q$-space $D/C$ has more than two elements.  
Let $d$ denote the minimum distance of $D$ with respect to $\mtt{P}$. 
If the inclusions $\{0\} \subset C \subsetneq D \subset \F_{q^t}^n$ hold, then the minimum distance of the (nonadditive, nonlinear) code $D\setminus C$ is $d$.
\end{Lemma}

\begin{proof}
Let $u$ and $v$ be two vectors from $D$ such that $\rho_{\mtt{P}}(u,v) = d$.
If both of the vectors $u$ and $v$ are elements of $D\setminus C$, then there is nothing to prove. 
We proceed with the assumption that $v\in C$.
If $u$ is an element of $C$ as well, then for every $w\in D\setminus C$ we have $\{u+w,v+w\}\subseteq D\setminus C$.
In this case, our assertion follows directly from the fact that $\rho_{\mtt{P}} ( u+w, v+w ) = \rho_{\mtt{P}}(u,v) =d$. 
Now it remains to show our claim in the case that $v\in C$ and $u\in D\setminus C$. 
In this case, we claim that there exists $w\in D\setminus C$ such that $u+w\in D\setminus C$.
Towards a contradiction, let us assume that for every $w\in D\setminus C$ we have $u+w\in C$.
Since there are three distinct cosets of $C$ in $D$, let us assume that $C, x+C$, and $u+C$ are distinct.
Then $w:=x-u$ is not an element of $C$.
It follows that $u+w \in D\setminus C$. 
This finishes the proof of our lemma.
\end{proof}

\section{Transfer}\label{S:Transfer}

We briefly review the well-known connection between the additive codes in $\F_{q}^{2n}$ 
and the quantum stabilizer codes in $(\C^q)^{\otimes n}$ that we mentioned in the introduction section. 
In fact, we are going to use this connection rather heavily in the rest of our paper, in particular in our proof of our Theorem~\ref{intro:T2} below. 

Let $Q$ be a stabilizer code with stabilizer group $S\leqslant G_n$. 
The elements of the error group $G_n$ are all of the form $\xi^c X(\mbf{a}) Z(\mbf{b})$, where $\xi$ is a fixed $p$-th root of unity, 
$c\in \{0,1,\dots, p-1\}$, $\mbf{a}=(a_1,\dots, a_n)\in \F_q^n$,  and $\mbf{b}=(b_1,\dots, b_n)\in\F_q^n$. 
We consider the map  
\begin{align}\label{A:usedinthenextresult}
\varphi : G_n &\longrightarrow \F_q^{2n} \\
\xi^c X(\mbf{a}) Z(\mbf{b}) &\longmapsto \underline{\mbf{a}\mbf{b}}. \notag
\end{align}
Then, $\varphi$ maps $S$ onto an additive code $C$, which is self-orthogonal with respect to the trace-symplectic form.
In fact, we have 
$$
\varphi( S) = \varphi (Z_{G_n} \rtimes S)= C,
$$
where $Z_{G_n}$ denotes the centralizer of $G_n$.
Finally, $\varphi$ maps the centralizer of $S$ in $G_n$ onto the trace-symplectic dual of $C$, that is, $C^{\perp_s}$. 
For all of these assertions, see the references~\cite[\S V]{AK2001} and~\cite[Theorem 13]{KKKS2006}.

In the context of Hamming metrics, a correspondence between stabilizer codes and additive codes had been previously established. 
Specifically, for binary codes, this connection was first identified by Calderbank, Rains, Shor, and Sloane in their seminal work~\cite{CRSS1997}. 
Subsequently, in the $q$-ary case, Ashikmin and Knill were the first to provide a formal proof~\cite{AK2001}, 
with further refinement and clarification introduced by Ketkar, Klappenecker, Kumar, and Sarvepalli in their comprehensive paper~\cite{KKKS2006}.
We will extend these results by using poset metrics. 

Let us recall the statement of our Theorem~\ref{intro:T2} from Introduction. 
\medskip

{\em 
Let $\mtt{P}$ be a poset on $[n]$. 
Let $Q$ be a stabilizer $\mtt{P}$-code.
Let $S$ denote the stabilizer group of $Q$ in $G_n$. 
Then there is a natural map $\varphi : G_n \to \F_q^{2n}$ sending $S$ onto a self-orthogonal additive code, $C\subseteq \F_q^{2n}$. 
Furthermore, the minimum distance of $Q$ is given by 
\[
\mtt{d}_{\mtt{P}} = \min \{wt_{\Delta\mtt{P}} (\underline{\mathbf{a}\mathbf{b}}) :\  \underline{\mathbf{a}\mathbf{b}} \in C^{\perp_s}\setminus C\},
\] 
where $\perp_s$ indicates the dual code with respect to the trace-symplectic form on $\F_q^{2n}$.
}

\begin{proof}[Proof of Theorem~\ref{intro:T2}]
Let $\varphi$ be the map defined in (\ref{A:usedinthenextresult}). 
As we mentioned earlier, we know that $\varphi(S) = \varphi(Z_{G_n} \rtimes S) = C$, where $C$ is a self-orthogonal additive code with respect to the trace-symplectic form. 
Moreover, $\varphi$ maps the centralizer of $S$ in $G_n$ onto the trace-symplectic dual, $C^{\perp_s}$. 
Now, the weight of an error operator $g:=\xi^c X(\mbf{a}) Z(\mbf{b}) $ is the number of its nonidentity tensor components.
It follows that the operator $\mtt{P}$-weight of $g$, which is defined as the cardinality of the lower order ideal that is generated by the indices of the nonidentity tensor components of $D$, is equal to the symplectic $\mtt{P}$-weight of $\underline{\mbf{a}\mbf{b}}$.
Recall also that the minimum distance of $Q$ is $\mtt{d}_\mtt{P}$ if and only if $Q$ detects all errors in $G_n$ of weight less than $\mtt{d}_\mtt{P}$, but cannot detect some error $g$ of weight $\mtt{d}_\mtt{P}$.
Therefore, the minimum of the operator $\mtt{P}$-weights of all error operators from $C_{G_n}(S)\setminus S$ gives the minimum distance of $Q$ with respect to the metric $d_{\Delta\mtt{P}}$. 
Hence, the proof follows. 
\end{proof}

\section{Singleton's Bound For Stabilizer Poset Codes}\label{S:HKS}

In this section we prove Theorem~\ref{intro:T1}.
Let us recall its statement for convenience. 
\medskip 

{\em Let $\mtt{P}$ be a poset on $[n]$.
If a stabilizer $\mtt{P}$-code has the parameters $[[n,K,\mtt{d}_\mtt{P}]]_q$, where $K>1$, then the following inequality holds: 
$K \leq q^{n-2\mtt{d}_{\mtt{P}}+2}$.}
\begin{proof}[Proof of Theorem~\ref{intro:T1}]
Let $Q$ be a stabilizer $\mtt{P}$-code with parameters $[[n,K,\mtt{d}_{\mtt{P}}]]_q$. 
Let $S$ denote the stabilizer group of $Q$. 
It follows from the orbit-stabilizer correspondence that $K = q^n / |S|$. 
Let $C$ be the additive code that is obtained as the image of $S$ under the map $\varphi$ of Theorem~\ref{intro:T2}.
Evidently, we have $|S| = |C|$. (For more details, see the proof of ~\cite[Theorem 13]{KKKS2006}.)

Let $A(z)$ and $B(z)$ denote the symplectic weight enumerators of $C$ and $C^{\perp_s}$, respectively. 
The following relationship between $A(z)$ and $B(z)$ was discovered in~\cite[Theorem 23]{KKKS2006}:
\begin{align}\label{A:AB}
B(z) = \frac{ (1+(q^2-1)z)^n}{|C|} A \left( \frac{ 1-z} {1+ (q^2-1)z }\right).
\end{align}
Notice that, setting $z=1$ in (\ref{A:AB}), gives us the cardinality, 
$$
B(1) = |C^{\perp_s}|.
$$ 
But on the right hand side of (\ref{A:AB}), we get $\frac{1}{|C|}A(0)$. 
Since $A(0)$ is the number of symplectic weight 0 vectors in $C$, which is equal to 1, we see that 
\begin{align*}
|C^{\perp_s}| = \frac{q^{2n}}{|C|} = q^n K.
\end{align*}

In Lemma~\ref{L:isometry}, we showed that there is an isometry between $(\F_{q}^{2n},\rho_{\Delta\mtt{P}})$ and $(\F_{q^2}^n,\rho_{\mtt{P}})$. 
Let $D^{\perp_a}$ (resp. $D$) denote the isometric image in $\F_{q^2}^n$ of $C^{\perp_s}$ (resp. $C$).
We apply the Hyun-Kim-Singleton Theorem (Lemma~\ref{L:HKSbound}) to the additive code $D^{\perp_a}$ in $\F_{q^2}^n$: 
\begin{align*}
|C^{\perp_s}| = |D^{\perp_a}| \leq q^{2n - 2d_{\mtt{P}}(D^{\perp_a})+2}.
\end{align*}
Notice that the number of elements of the quotient $p$-vector space $D^{\perp_a}/D$, equivalently, the number of elements of the $p$-vector space $C^{\perp_s}/C$, is given by $K^2$. 
Therefore, we see from Lemma~\ref{L:morethantwo} that, whenever $K>1$, we have 
\begin{align*}
q^{2n - 2d_{\mtt{P}}(D^{\perp_a})+2} =q^{2n- 2d_{\mtt{P}}(D^\perp_a\setminus D)+2}. 
\end{align*}
Since we have $d_{\mtt{P}}(D^\perp_a\setminus D)= d_{\Delta\mtt{P}}(C^{\perp_s} \setminus C)$, we obtain
\begin{align}\label{applyHKS}
q^n K = |C^{\perp_s}| \leq q^{2n- 2d_{\mtt{P}}(D^\perp_a\setminus D)+2} = q^{2n- 2d_{\Delta\mtt{P}}(C^{\perp_s} \setminus C)+2}. 
\end{align}
Now it follows from the definition of the minimum distance of a stabilizer code with respect to $\mtt{P}$ that the right hand side of (\ref{applyHKS}) is equal to $q^{2n- 2\mtt{d}_{\mtt{P}}(Q)+2}$.
This finishes the proof of our theorem.
\end{proof}

\section{MDS Property}\label{S:MDS1}

The purpose of this section is to give a proof of Theorem~\ref{intro:T3}.
We first recall its statement.
\medskip

We fix a poset $\mtt{P}$ on $[n]$. 
Let $S$ denote the stabilizer group of a stabilizer $\mtt{P}$-code $Q$.
Let $D$ denote the self-orthogonal additive code in $\F_q^{2n}$ corresponding to $S$, 
which is given by the image of the composition $\psi^{-1} \circ \phi : S \to \F_{q^2}^n$, where $\phi$ is defined in Theorem~\ref{intro:T2} and 
$\psi: \F_q^{2n}\to \F_{q^2}^n$ is the isometry that is defined before Theorem~\ref{intro:T3} in the introduction.  
In this notation, our Theorem~\ref{intro:T3} states the following.
\medskip

{\em 
Let $Q$ be a stabilizer $\mtt{P}$-code with parameters $[[n,K,\qmd{P}]]_q$. 
Let $D=\psi^{-1}\circ \varphi (S)$, where $\psi$ is defined above and $\varphi$ is the map from Theorem~\ref{intro:T2}. 
If $Q$ is $\mtt{P}$-pure, then the following assertions hold: 
\begin{enumerate}
\item[(1)] $Q$ is an MDS stabilizer $\mtt{P}$-code if and only if $D^{\perp_a}$ is an MDS additive $\mtt{P}$-code.
\item[(2)] If $Q$ is an MDS stabilizer $\mtt{P}$-code, and the inequality $n-\mtt{d}_{\mtt{P}} +2 \leq d_{\mtt{P}}(D)$ holds, then $D$ is an MDS additive $\mtt{P}$-code. 
\end{enumerate}}

\begin{proof}[Proof of Theorem~\ref{intro:T3}]

(1) 
We assume that $Q$ is a $\mtt{P}$-pure MDS stabilizer $\mtt{P}$-code. 
Let $C$ denote the additive code in $\F_q^{2n}$ such that $\psi^{-1}(C) = D$. 
It follows readily from Theorem~\ref{intro:T2} and our assumption of $Q$ being $\mtt{P}$-pure that $C^{\perp_s}$ is $\Delta\mtt{P}$-pure.
Since $Q$ is an MDS stabilizer code, we know that $K=q^{n-2\mtt{d}_{\mtt{P}}+2}$. 
Recall from the proof of Theorem~\ref{intro:T1} that $|C| = q^n/K$ and $|C^{\perp_s}| = q^nK$. 
It follows that $|C^{\perp_s}| = q^{2n-2\mtt{d}_{\mtt{P}}+2}$.
Since $|D^{\perp_a}|=|C^{\perp_s}|$ and $\mtt{d}_{\mtt{P}} = d_{\Delta\mtt{P}}(C^{\perp_s})= d_{\mtt{P}}(D^{\perp_a})$, we obtain the following equalities: 
\begin{align*}
|D^{\perp_a}| &= |C^{\perp_s}| \\
&= q^{2n-2\mtt{d}_{\mtt{P}}+2} \\ 
&= q^{2n - 2 d_{\mtt{P}}(D^{\perp_a}) +2} \\
&= (q^2)^{n-d_{\mtt{P}}(D^{\perp_a})+1}.
\end{align*}
Since $D^{\perp_a}$ is a code in $\F_{q^2}^n$, the last equality shows that $D^{\perp_a}$ is an MDS additive $\mtt{P}$-code.

Conversely, if $D^{\perp_a}$ is an MDS additive $\mtt{P}$-code, then $C^{\perp_s}$ is an MDS code with respect to the symplectic $\mtt{P}$-weight.
In other words, we have 
$$
| C^{\perp_s} |= q^{2(n-d_{\Delta\mtt{P}}(C^{\perp_s}) +1)}.
$$
But $|D^{\perp_a}|=|C^{\perp_s}|=q^nK$ implies that $K=q^{n-2d_{\Delta\mtt{P}}(C^{\perp_s}) +2}$.
By Theorem~\ref{intro:T2}, we see that $d_{\Delta\mtt{P}}(C^{\perp_s}) \leq \mtt{d}_{\mtt{P}}$.
At the same time, since $Q$ is $\mtt{P}$-pure, we know the inequality $\mtt{d}_{\mtt{P}}\leq d_{\Delta\mtt{P}}(C^{\perp_s})$.
Therefore, we obtain the equality $\mtt{d}_{\mtt{P}}= d_{\Delta\mtt{P}}(C^{\perp_s})$.
In conclusion, the dimension $K$ of the stabilizer code satisfies identity $K=q^{n-2\mtt{d}_{\mtt{P}} +2}$.
This implies that $Q$ is an MDS stabilizer $\mtt{P}$-code, finishing the proof of part (1) of our theorem.

(2) Our assumption that $Q$ is an MDS stabilizer $\mtt{P}$-code implies that $K=q^{n-2\mtt{d}_{\mtt{P}} +2}$. 
Hence, we deduce, by using the equalities $|D| = |C| = q^n/K$, that $|D| = q^{2\mtt{d}_{\mtt{P}}-2}= (q^2)^{d_{\mtt{P}}(D)-1}$.
Then, by Lemma~\ref{L:HKSbound} (Hyun-Kim-Singleton bound), we see that $\mtt{d}_{\mtt{P}} -1 \leq n- d_{\mtt{P}}(D)+1$.
As part of our hypothesis we have the inequality, $n-\mtt{d}_{\mtt{P}} +2\leq d_{\mtt{P}}(D)$. 
By combining these two inequalities, we obtain 
\begin{align*}
  \mtt{d}_{\mtt{P}} -1 = n-d_{\mtt{P}}(D)+1.
\end{align*}
Of course, this equality yields at once that $|D|=(q^2)^{n-d_{\mtt{P}}(D)+1}$. 
In other words, $D$ is an MDS additive $\mtt{P}$-code.
This finishes the proof of part (2). Hence, our theorem follows.
\end{proof}

In the rest of this section we will prove our second characterization of MDS stabilizer codes in terms of $\mtt{I}$-perfect additive codes. 
We recall our Theorem~\ref{intro:T5}.
\medskip

{\em 
Let $Q$ be a $\mtt{P}$-pure stabilizer code $Q$ with parameters $[[n,K,\qmd{P}]]_q$.
Then $Q$ is an MDS stabilizer $\mtt{P}$-code if and only if $Q$ is an $\mtt{I}$-perfect stabilizer code for every ideal $\mtt{I}$ of size $n-\log_q |D^{\perp_a}|$ 
and the number $\log_q K$ is an integer.
}

\begin{proof}[Proof of Theorem~\ref{intro:T5}]
Since $Q$ is a $\mtt{P}$-pure stabilizer code, 
it follows from Theorem~\ref{intro:T3} that $Q$ is an MDS stabilizer $\mtt{P}$-code if and only if $D^{\perp_a}$ is an MDS additive $\mtt{P}$-code.
Hence, our task is reduced to investigating the $\mtt{I}$-perfectness of $D^{\perp_a}$.
But our Theorem~\ref{T:Additive_MDS_codes} states that 
$D^{\perp_a}$ is an additive MDS $\mtt{P}$-code if and only if $\log_q |D^{\perp_a}|$ is an integer, and $D^{\perp_a}$ is an $\mtt{I}$-perfect code for every ideal $\mtt{I}$ of size $n-\log_q |D^{\perp_a}|$. 
We already pointed out that $|D^{\perp_a}| = K q^n$.  
Now, notice that $\log_q |D^{\perp_a}|$ is an integer if and only if $\log_q K$ is an integer. 
Hence, we the following assertions are equivalent: 
\begin{itemize}
\item $Q$ is an MDS stabilizer $\mtt{P}$-code.
\item $D^{\perp_a}$ is an additive MDS $\mtt{P}$-code.
\item The following two conditions are satisfied:
\begin{enumerate}
\item[(1)] $\log_q K$ is an integer, and
\item[(2)] $D^{\perp_a}$ is an $\mtt{I}$-perfect code for every ideal $\mtt{I}$ of size $n-\log_q |D^{\perp_a}|$. 
\end{enumerate} 
\end{itemize}
This finishes the proof of our theorem.
\end{proof}

\section{Final Remarks}\label{S:MDS2}

In this section, we discuss the existence of MDS stabilizer $\mtt{P}$-codes. 
The classical version of this discussion was established by Hyun and Kim~\cite{HyunKim2008} who showed that for every classical code $D$ there exists an appropriate poset $\mtt{P}$ such that $D$ is an MDS $\mtt{P}$-code.

Let $D$ be an $\F_q$-linear $\F_{q^t}$-code of length $n$, dimension $k$. 
Let $G$ be a $k\times n$ generator matrix for $D$ as an $\F_q$-linear code. 
Clearly, the entries of each column of $G$ span an at most $t$-dimensional $\F_q$-vector subspace of $\F_{q^t}$. 
Let $k_1$ denote the dimension of the $\F_q$-vector subspace spanned by the entries of the first column of $G$. 
By applying elementary row operations to $G$, we define a new generator matrix, denoted $G_1$, such that 
\begin{enumerate}
\item the first $k_1$ entries of the first column of $G_1$ form an $\F_q$-basis for the $\F_q$-span of all entries of the first column of $G$. 
\item The remaining entries of the first column of $G_1$ are 0. 
\end{enumerate}
We repeat this process on the matrix $G_1'$ that is obtained from $G_1$ by removing its first column and first $k_1$ rows.
We construct a new matrix $G_2$ such that the first $k_2$ entries of the first column of $G_2$ form an $\F_q$-basis for the $\F_q$-span of all entries of the first column of $G_1'$, and the remaining entries in the first column of $G_2$ are all 0.  
Then we continue inductively to construct a sequence $k_1,k_2,\dots, k_r$ of {\em row reduction numbers} and a {\em generator matrix $G'$ in reduced form} of $D$.
The following properties are evident from the algorithmic construction:
\begin{itemize}
\item $0\leq k_i \leq t$ for every $i\in [r]$;
\item $k_r \neq 0$; 
\item $k_1+\cdots + k_r = k$;
\item the $k_i$ entries in rows $k_1+\cdots + k_{i-1}+1$ through $k_1+\cdots + k_i$ of column $i$ of $G'$ are $\F_q$-linearly independent elements of $\F_{q^t}$;
\item in the $i$th column of $G'$, the entries in rows $k_1+\cdots + k_i+1$ through $k$ are all 0. 
\end{itemize}

Let $s$ denote $\left\lceil \frac{k}{t} \right\rceil$, where $k$ denotes the $\F_q$-vector space dimension of $D$. 
Let $d$ denote the minimum distance of $D$ with respect to antichain poset $\mtt{P}_0$. 
In~\cite[Theorem 10]{Huffman2013}, by investigating the generator matrix $G'$ in reduced form for $D$, 
Huffman explained the following version of the Singleton's bound for additive codes: 
\begin{align}\label{T:Huffman2013}
s \leq n - d +1.
\end{align}
Here, $d$ is the minimum distance of $D$ with respect to the Hamming metric. 
We proceed to show that (\ref{T:Huffman2013}) can be seen as a special case of Lemma~\ref{L:HKSbound} after replacing $\F_q$ with $\F_{q^t}$.

\begin{Proposition}\label{P:additiveMDS}
Let $\mtt{P}$ be a poset on $[n]$. 
Let $D$ be an $\F_q$-linear $\F_{q^t}$-code of length $n$ with $\F_q$-dimension $k$. Let $s:=\left\lceil \frac{k}{t} \right\rceil$.
Then we have $d_{\mtt{P}}(D)\leq n-s +1$.
In particular, if $\mtt{P}$ is the antichain poset, then we have $d \leq n-s+1$, where $d$ is the minimum distance of $D\subset \F_{q^t}^n$ with respect to the Hamming metric.
\end{Proposition}
\begin{proof}
We apply Lemma~\ref{L:HKSbound} to $D$:
$$
|D| \leq (q^t)^{n-d_{\mtt{P}}(D)+1}.
$$
Since $|D|=q^k$, by taking the logarithm of both sides of this inequality, we find that  
\begin{align}\label{HSfollowsfromHKS}
\frac{k}{t} \leq n-d_{\mtt{P}} +1. 
\end{align}
Since the right hand side of (\ref{HSfollowsfromHKS}) is an integer, we can replace the rational number $\frac{k}{t}$ by its ceiling, which is $s$.
This finishes the proof of the first assertion of our proposition. 
For our second assertion, we specialize $\mtt{P}$ to the antichain poset $\mtt{P}_0$.
Hence, we recover the inequality in (\ref{T:Huffman2013}).
This finishes the proof of our proposition.
\end{proof}

\begin{Corollary}\label{C:additiveMDS}
Let $\mtt{P}$ be a poset on $[n]$. 
Let $C$ be a linear code in $\F_q^{2n}$ such that $\dim_{\F_q} C= k$. Let $s:=\left\lceil \frac{k}{2} \right\rceil$.
Then we have $d_{\Delta\mtt{P}}(C)\leq n-s +1$.
\end{Corollary}
\begin{proof}
We use the isometry $\psi$ from Lemma~\ref{L:isometry}.
Let $D$ denote the additive $\F_q$-linear $\F_{q^2}^n$-code defined by $\psi^{-1}(C)$. 
Then Proposition~\ref{P:additiveMDS} implies that $d_{\mtt{P}}(C)\leq n-s+1$.
Since $\psi$ is an isometry, the parameters of $D$ satisfies the same inequality: $d_{\Delta \mtt{P}}(D) \leq n-s+1$.
This finishes the proof.
\end{proof}

We now prove our Theorem~\ref{intro:T4}.  
Let us recall its statement for convenience. 
\medskip

{\em 
Let $D$ be an $\F_{q^2}$-linear code in $\F_{q^2}^n$ such that $D^{\perp_a} \subseteq D$.
If $k$ denotes the dimension $\dim_{\F_{q^2}} D$, then there exists a poset $\mtt{P}$ on $[n]$ and a $\mtt{P}$-pure MDS stabilizer $\mtt{P}$-code $Q$ with parameters $[[ n, n-2k, \mtt{d}_{\mtt{P}}] ]_q$.
}

\begin{proof}[Proof of Theorem~\ref{intro:T4}]
The inclusion $D^{\perp_a} \subseteq D$ implies that, for every poset $\mtt{P}$ on $[n]$, the weights of $D$ and $D^{\perp_a}$ satisfy the inequality  $d_{\mtt{P}}(D) \leq d_{\mtt{P}}(D^{\perp_a})$.
Hence, $D$ is a $\mtt{P}$-pure additive code. 
Now, by~\cite[Theorem 4.2]{HyunKim2008}, we know that there exists a poset $\mtt{P}$ such that $D$ is an MDS $\mtt{P}$-code. 
Let $Q$ denote the stabilizer code determined by $D$. 
By Remark~\ref{R:Pure}, since $D$ is $\mtt{P}$-pure, so is $Q$. 
We now apply Theorem~\ref{intro:T3} (1).
It implies that $Q$ is an MDS stabilizer $\mtt{P}$-code.
It remains to compute the parameters of $Q$. 
Since $Q$ is MDS, we know that the dimension $K$ of $Q$ satisfies:
$$
q^n/|D| = K = q^{n - 2d_{\mtt{P}}(Q)+2}.   
$$
It follows that $|D| = q^{2d_{\mtt{P}}(Q)-2}$.
Since $D$ is an $\F_{q^2}$-linear code of dimension $k$, we see that 
$d_{\mtt{P}}(Q) = k+1$. 
Hence, the proof of our theorem follows. 
\end{proof}

\begin{Remark}
It is known that there is a close relationship between the trace-alternating form and the {\em Hermitian inner product} defined by $\langle \mbf{a}, \mbf{b} \rangle_h:= \mbf{a}^q\cdot \mbf{b}$ where $\mbf{a},\mbf{b}\in \F_{q^2}^n$. 
Here, $\mbf{a}^q$ is the vector obtained from $\mbf{a}$ by raising each entry to its $q$-th power.
For linear subspaces of $\F_{q^2}^n$, the Hermitian inner product dual of a subspace $D\subseteq \F_{q^2}^n$ is equal to the trace-alternating form dual of $D$. In other words, we have $D^{\perp_h} = D^{\perp_a}$.
Hence, our Theorem~\ref{intro:T4} shows that we can create an MDS stabilizer poset code for any isotropic subspace of $\F_{q^2}^n$ with respect to Hermitian inner product. 
\end{Remark}

\section*{Acknowledgements}

This research was conducted while the author was visiting the Okinawa Institute of Science and Technology (OIST) through the Theoretical Sciences Visiting Program (TSVP).
The author gratefully acknowledges the Louisiana Board of Regents grant, contract no. LEQSF(2023-25)-RD-A-21.

\bibliographystyle{plain}
\bibliography{References.bib}
\end{document}